\documentclass[sigconf, nonacm]{acmart}
\pdfoutput=1
\newcommand\vldbdoi{10.14778/3476249.3476263}
\newcommand\vldbpages{XXX-XXX}
\newcommand\vldbvolume{14}
\newcommand\vldbissue{11}
\newcommand\vldbyear{2021}
\newcommand\vldbauthors{\authors}
\newcommand\vldbtitle{\shorttitle} 
\newcommand\vldbavailabilityurl{https://github.com/AwesomeYifan/learning-based-set-sim-search}
\newcommand\vldbpagestyle{empty}
\usepackage{subfig}
\usepackage{booktabs}
\usepackage{graphicx}
\usepackage{balance}
\usepackage{hhline}
\usepackage{float}
\usepackage{multirow}
\usepackage{amsmath}
\usepackage{algorithm}
\usepackage{algorithmicx}
\usepackage{graphicx}
\usepackage[utf8]{inputenc}
\usepackage[noend]{algpseudocode}
\usepackage{bbm}
\usepackage[skip=0pt]{caption}

\newcommand{\les}{$\mathsf{LES}^3$}
\newcommand{\cS}{\mathcal{D}}

\begin{document}
\title{LES{$^3$}: Learning-based Exact Set Similarity Search}


\author{Yifan Li}
\affiliation{%
  \institution{York University}
}
\email{yifanli@eecs.yorku.ca}

\author{Xiaohui Yu}
\affiliation{%
  \institution{York University}
}
\email{xhyu@yorku.ca}

\author{Nick Koudas}
\affiliation{%
  \institution{University of Toronto}
}
\email{koudas@cs.toronto.edu}

\begin{abstract}
Set similarity search is a problem of central interest to a wide variety of applications such as data cleaning and web search. Past approaches on set similarity search utilize either heavy indexing structures, incurring large search costs or indexes that produce large candidate sets. In this paper, we design a learning-based exact set similarity search approach, \les. Our approach first partitions sets into groups, and then utilizes a light-weight bitmap-like indexing structure, called token-group matrix (TGM), to organize groups and prune out candidates given a query set. In order to optimize pruning using the TGM, we analytically investigate the optimal partitioning strategy under certain distributional assumptions. Using these results, we then design a learning-based partitioning approach called L2P and an associated data representation encoding, PTR, to identify the partitions. We conduct extensive experiments on real and synthetic datasets to fully study \les, establishing the effectiveness and superiority over other applicable approaches.
\end{abstract}

\maketitle

\pagestyle{\vldbpagestyle}
\begingroup\small\noindent\raggedright\textbf{PVLDB Reference Format:}\\
\vldbauthors. \vldbtitle. PVLDB, \vldbvolume(\vldbissue): \vldbpages, \vldbyear.\\
\href{https://doi.org/\vldbdoi}{doi:\vldbdoi}
\endgroup
\begingroup
\renewcommand\thefootnote{}\footnote{\noindent
This work is licensed under the Creative Commons BY-NC-ND 4.0 International License. Visit \url{https://creativecommons.org/licenses/by-nc-nd/4.0/} to view a copy of this license. For any use beyond those covered by this license, obtain permission by emailing \href{mailto:info@vldb.org}{info@vldb.org}. Copyright is held by the owner/author(s). Publication rights licensed to the VLDB Endowment. \\
\raggedright Proceedings of the VLDB Endowment, Vol. \vldbvolume, No. \vldbissue\ %
ISSN 2150-8097. \\
\href{https://doi.org/\vldbdoi}{doi:\vldbdoi} \\
}\addtocounter{footnote}{-1}\endgroup

\ifdefempty{\vldbavailabilityurl}{}{
\vspace{.3cm}
\begingroup\small\noindent\raggedright\textbf{PVLDB Artifact Availability:}\\
The source code, data, and/or other artifacts have been made available at \url{\vldbavailabilityurl}.
\endgroup
}

\section{Introduction}
Given a database $\cS$ of sets each comprised of tokens (a token can be an arbitrary string from a given alphabet $\Sigma$, a unique identifier from a known domain, etc.), a single query set $Q$ (consisting of tokens from the same domain), and a similarity measure $Sim(*)$, the problem of {\em set similarity search} is to identify from $\cS$ those sets that are within a user defined similarity threshold to the query $Q$ (range query) or $k$ sets that are the most similar to $Q$ ($k$NN query). This operation is essential to a wide spectrum of applications, such as {data cleaning \cite{hadjieleftheriou2008hashed,wang2019uni}, data integration \cite{dong2018data,ge2019speculative}, query refinement \cite{sahami2006web}, and digital trace analysis \cite{li2019top}}.  For example, a common task in data cleaning is to perform approximate string matching to identify near duplicates of a given query string. When strings are tokenized, the task of approximate string matching becomes a set similarity search problem. Given its prevalent use, efficient set similarity search is of paramount importance.  A brute-force approach to supporting set similarity search is to scan all the sets in  $\cS$ and evaluate $Sim(*)$ between $Q$ and each set in $\cS$ to obtain the results. When $\cS$ is large or such operations are carried out repeatedly, however, its efficiency becomes a major concern.

Existing  proposals to improve the search performance adopt a filter-and-verify framework: in the filter step, candidate sets are generated based on indexes on $\cS$, and the candidate sets are further examined, computing the similarity between $Q$ and each candidate set in the verify step. {Depending on the indexes used in the filter step, existing methods can be categorized into two groups: inverted index-based and tree-based. Inverted index-based methods build inverted index on tokens and only fetch those sets containing (a subset of) tokens present in the query set as candidates. Tree-based methods \cite{zhang2017efficient,DBLP:journals/tkde/ZhangWWX20} transform sets to scalars \cite{zhang2017efficient} or vectors \cite{DBLP:journals/tkde/ZhangWWX20} and insert them into B+-trees or R-trees, which are then used at query processing time to quickly identify the candidate sets. As verification of a candidate set can be done very efficiently under almost all well-known set similarity measures (e.g., Jaccard, Dice, Cosine similarity) incurring a cost linear in the size of the set, optimization of the filter step is critical.  Unfortunately, existing methods either utilize heavy-weight indexes that incur expensive storage consumption and excessive scanning cost during filtering \cite{DBLP:journals/tkde/ZhangWWX20}, or employ indexes that are light-weight but with very limited pruning efficiency leading to an overly large candidate set \cite{zhang2017efficient}.  Therefore, existing approaches mostly do not solve the set similarity search problem effectively. In fact for realistically low similarity thresholds or large result sizes, as we demonstrate in our experiments, the brute-force approach may perform much better.}

In this paper, we study the problem of set similarity search, and propose a new approach named \les\ (short for \underline{L}earning-based \underline{E}xact \underline{S}et \underline{S}imilarity \underline{S}earch) that strives to reduce the time needed for filtering and increase the pruning efficiency of the index structure at the same time. At a high level, our approach also adopts a filter-and-verify framework; however we advocate the partitioning of the sets in $\cS$ into non-overlapping {\em groups} for filtering. {What differentiates our approach from existing methods is that instead of building complex index structures that could become too expensive to utilize at run-time, we introduce a light-weight index structure called {\em token-group matrix} (TGM); this structure is essentially a collection of bit-maps, to organize all groups, yielding comparable or higher pruning efficiency with only a fraction of storage cost and thus highly scalable.} The TGM captures the association between tokens and groups, and allows us to quickly compute an upper bound on the similarity between the query set $Q$ and any set in a given group. Such upper bounds can then be used for {\em pruning} unrelated groups and directing search to the most promising groups. 

As the search efficiency relies on the pruning efficiency of the TGM which in turn depends on how well the sets are partitioned, we formulate the construction of TGM as an optimization problem, that aims to identify the partitioning of sets that yields the highest pruning efficiency. We first analytically model the base case in which every token has the same probability of appearing in any set. Our developments reveal that the optimal partitioning has two properties: balance and intra-group coherence. We then design a {\em general partitioning objective} ($GPO$) that strives to maximize the pruning efficiency, taking both properties into consideration. 

We showcase that the optimal partitioning is NP-Hard and explore the use of algorithmic and machine learning-based methods to solve the optimization problem. Recent works \cite{kraska2018case,galakatos2019fiting} have demonstrated that machine learning techniques have solid performance in learning the cumulative distribution function (CDF) of real data sets and this property can be used in important data management tasks such as indexing \cite{li2020lisa} and sorting \cite{kristo2020case}. We establish that machine learning techniques can also be utilized to produce superior solutions to hard optimization problems, central to other important indexing tasks such as those in support of set similarity search.

{Complementary to existing works \cite{li2020lisa,galakatos2019fiting} that utilize models such as piece-wise linear regression to learn a CDF, we explore  models that are much better a fit, proposing a unique ensemble learning method suitable for progressive partitioning in our setting.}
The main difficulty in solving the optimization problem is that depending on $\cS$, the number of groups needed for effective pruning can be large, so it is highly challenging to train a single network that would place any given set into one of these groups. 
As such, we propose a new learning framework named L2P (short for \underline{L}earning to \underline{P}artition) to address this challenge. L2P trains a cascade of Siamese networks to hierarchically partition the database $\cS$ into increasingly finer groups until the desired number of groups is reached, resulting in $2^i$ groups at level $i$. The loss function for the Siamese network is specifically designed to minimize the distances between sets in the same group. As the input of a Siamese network has to be a vector, we devise a novel and efficient set representation method, {\em path-table representation} (PTR), that specifically caters to the needs of our optimization problem and proves to be a better fit than applicable embedding techniques. {Although training ML models is known to be time-consuming, as will be shown in Section \ref{sec:exp}, L2P yields better partitioning results with much shorter processing time and only a small fraction of memory usage compared with other widely-adopted partitioning methods.}

We fully develop the query processing algorithms for both range search and $k$NN search based on the TGM, and conduct extensive experiments on synthetic and real data sets to study the properties of our proposal and compare it against other applicable approaches. Our results demonstrate that both the proposed set representation method and the learning framework lead to much stronger pruning efficiency than competing methods. Overall, the proposed \les\ method significantly outperforms the baseline methods in both memory-based and disk-based settings.

In summary, we make the following main contributions.
\begin{itemize}
\item We propose a learning-based approach, \les, for exact set similarity search, which partitions the database into groups to facilitate filtering. Central to \les\ is TGM, a light-weight yet highly effective index that provides stronger pruning efficiency with less cost than state-of-the-art indexes. 
\item We formally analyze the partitioning of the database into groups, casting it as an optimization problem and discussing its distinction from well-studied clustering problems. 

\item We devise a novel learning framework, L2P, to solve the partitioning optimization problem, { which yields significantly better partitioning results while incurring a small fraction of processing time and space cost compared with traditional algorithmic methods}. L2P consists of a cascade of Siamese networks, an architecture that is able to effectively learn a partition of the dataset at different granularities, with up to thousands of groups at the finest level. 
\item We develop a  carefully designed method for set representation, PTR, taking group separation into consideration. PTR theoretically and experimentally facilitates the training of L2P. { Compared with other embedding techniques, PTR is orders of magnitude faster in computing set representations, and thus is more suitable for the target application where millions or billions of sets are involved (see Section \ref{sec:exp}).}
\item We experimentally study the performance of \les, L2P, and PTR, varying parameters of interest, including the network structure, number of groups and result size. We also examine the scalability of \les\ utilizing real world large datasets in addition to previously used set similarity benchmarks. The proposed methods significantly and consistently outperform competing methods across a large variety of settings, {providing up to 5 times faster query processing and requiring up to $90\%$ less space in typical scenarios.}
\end{itemize}

The rest of the paper is organized as follows. In Section 2, we define the terminologies to be used throughout the paper. Section 3 introduces the index structure, TGM. In Section 4, we discuss how the partitioning problem can be formulated as an optimization problem. In Section 5, we propose a machine learning framework to solve the optimization problem. Section 6 presents the query processing algorithms for set similarity search. The experimental evaluation is presented in Section 7. Section 8 discusses related work, and Section 9 concludes this paper.  
\section{Preliminaries}

A {\em set} is an unordered collection of elements called {\em tokens} ({we also consider multiset which may contain duplicate tokens in the paper}).  We use $S$ to denote an arbitrary set and $t$ an arbitrary token. The database $\cS$ is a collection of sets, and all tokens form the token universe $\mathcal{T}$. Two sets are considered similar if the overlap in their tokens exceeds a user-defined threshold. Usually, such overlap is normalized to account for the size difference between sets. Examples of such similarity measures include Jaccard, Dice, and Cosine similarity. To make our discussion more concrete, we focus on Jaccard similarity, and discuss how our approach can be applied to other similarity measures in Section~\ref{sec:applicability}. Next we give the formal problem definitions.

\begin{definition}\textbf{\textit{k}NN Search}. Given the database of sets $\cS$, a set similarity measure $Sim(*)$, a query set\footnote{Without loss of generality, we assume throughout the paper that a query set consists of tokens existing in $\mathcal{T}$ only.  The case of query sets containing tokens not in $\mathcal{T}$ can be handled similarly and is discussed in Section \ref{sec:index}.} $Q$, and a result size $k$, find a collection $\mathcal{R}_Q^{k}\subseteq\cS$ s.t. $\vert \mathcal{R}_Q^k\vert=k$ and $\forall S\in\mathcal{R}_Q^k$, $\forall S'\in\cS-\mathcal{R}_Q^k$, $Sim(Q,S)\ge Sim(Q,S')$.
\end{definition}
\begin{definition}\textbf{Range Search}. Given the database of sets $\cS$, a set similarity measure $Sim(*)$, a query set $Q$, and a threshold $\delta$, find a collection $\mathcal{R}_Q^{\delta}\subseteq\cS$ s.t. $\forall S\in\mathcal{R}_Q^{\delta}$, $Sim(Q,S)\ge\delta$, and $\forall S'\in\cS-\mathcal{R}_Q^{\delta}$, $Sim(Q,S')<\delta$.
\end{definition}

Our goal is to accelerate the process of identifying the result collection $\mathcal{R}_Q^{k}$ or $\mathcal{R}_Q^{\delta}$ for the given $k$ or $\delta$. In general, the query answering process consists of a filtering step (choosing candidate sets) and a verification step (comparing candidate sets with the query set). The cost of the verification step depends directly on the pruning efficiency of the search process, which measures the proportion of sets in $\cS$ being pruned in the filtering step. 
\begin{definition}\textbf{Pruning Efficiency (PE)}. Let $\mathcal{S}_Q$ be the collection of candidate sets for which the similarities to $Q$ must be computed in the process of identifying $\mathcal{R}_Q^{k}$ or $\mathcal{R}_Q^{\delta}$. Then the pruning efficiency of query processing, denoted as $PE$, is $\frac{\vert\mathcal{D}\vert-(\vert\mathcal{S}_Q\vert-k)}{\vert\mathcal{D}\vert}$ for $k$NN query, or $\frac{\vert\mathcal{D}\vert-(\vert\mathcal{S}_Q\vert-\vert\mathcal{R}_Q^{\delta}\vert)}{\vert\mathcal{D}\vert}$ for range query.
\end{definition}

Clearly $PE$ falls in the range $[0,1]$. All other things being equal, a higher $PE$ leads to a lower verification cost. Our focus in this paper is therefore to design an approach for set similarity search that enjoys high PE and low filtering and verification cost.
\section{Token-Group Matrix}
\label{sec:token-group matrix}
The basic idea of our approach is to partition the sets in $\cS$ into non-overlapping groups and index them properly, so that the search space can be pruned (i.e., certain groups can be quickly eliminated from further consideration) to speed up query processing. At the heart of our proposal is the token-group matrix (TGM), the index that records the relationship between tokens and the groups resulting from partitioning. In this section, we present the index structure and discuss its applicability across different similarity measures. 
\subsection{Index Structure}
\label{sec:index}
Assume for now, that $\cS$ is already partitioned into $n$ non-overlapping groups, $\mathcal{G}_1,\cdots \mathcal{G}_n$; we defer the discussion of the strategies for partitioning to the next section. The goals of the index are simplicity (so that it incurs little computational and storage overhead) and effectiveness (providing high pruning efficiency). To this end, 
the TGM,  $M$, with size $n*\vert\mathcal{T}\vert$, is constructed in the following way:

\begin{footnotesize}
\begin{equation}
    M[g,t]=
    \begin{cases}
      1, & \text{if}\ \exists S\in\mathcal{G}_g\ \text{s.t.}\  t\in S\\
      0, & \text{otherwise}
    \end{cases}
    \quad
\end{equation}
\end{footnotesize}
where $t\in [1,\vert\mathcal{T}\vert]$ and $g\in [1,n]$.

An example of TGM is given in Figure~\ref{fig:tgm}, where $\mathcal{T}=\{A,B,C,D\}$ and six sets are partitioned into two groups $\mathcal{G}_0$ and $\mathcal{G}_1$.

\begin{figure}[htp]
\centering
\includegraphics[width=2.2in]{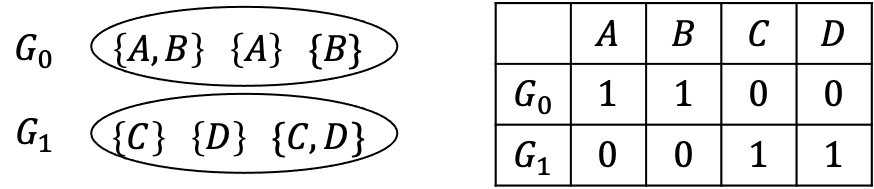}
\caption{An example of TGM}
\vspace{-0.15in}
\label{fig:tgm}
\end{figure}
The design of the TGM is based on the observation that when deciding whether a group of sets is a candidate for a query set or not, the only information needed is the number of common tokens they share. Such information can be easily obtained by visiting some elements in $M$, and thus we can compute a {\em similarity upper bound} between a query set $Q$ and a group of sets $\mathcal{G}_g$, which are useful in pruning the search space, as follows: 

\begin{footnotesize}
\begin{equation}
    \label{eq:ub}
    UB(Q,\mathcal{G}_g)=\frac{\sum_{t\in Q}M[g,t]}{\vert Q\vert}
\end{equation}
\end{footnotesize}

Continuing the example above, we assume that the query set is $\{A\}$ and $\mathcal{G}_0$ and $\mathcal{G}_1$ in Figure \ref{fig:tgm} are candidates. Then the similarity bound between the query set and $\mathcal{G}_0$ is $\frac{M[\mathcal{G}_0,A]}{\vert\{A\}\vert}=1$, and the upper bound for $\mathcal{G}_1$ is $\frac{M[\mathcal{G}_1,A]}{\vert\{A\}\vert}=0$.

{Although we assume $Q$ contains tokens in $\mathcal{T}$ only, the case where this does not hold can be handled by letting $M[*,t']=0$ for $t'\notin\mathcal{T}$ in Equation (\ref{eq:ub}). No further changes are required.}



In the query processing step, if the upper bound of group $\mathcal{G}_g$ exceeds a threshold (can be $\delta$ in range query, or the minimal $k$NN similarity found so far), we compare all sets in $\mathcal{G}_g$ with the query set. The time complexity of computing the similarity bounds between the query set and all groups of sets is $O(n\vert Q\vert)$. It in general costs much less than computing the similarity between the query set and each set in $\cS$, as  the number of groups is usually orders of magnitude smaller than $|\cS|$.  

In terms of space consumption, each element in TGM is represented by a bit, and TGM is essentially a bitmap index. It is evident that $M$ is usually a very sparse matrix as each set usually contains a very small portion of the tokens from the universe. When necessary, many existing compression techniques \cite{salomon2004data,salomon2010handbook} can be employed to reduce the size of $M$.


\subsection{Applicability}\label{sec:applicability}
Although Equation (\ref{eq:ub}) is computed assuming Jaccard index as the similarity metric, TGM works with many other set similarity measures as well, including measures that do not follow the triangle inequality, such as cosine similarity.

{
\begin{theorem}
\label{theorem:applicable}
For $\forall Q,S\subseteq\mathcal{T}$, let $R=Q\cap S$. TGM is applicable to set similarity search tasks with measure $Sim(*)$ if
\begin{enumerate}
    \item $Sim(Q,R)\ge Sim(Q,S)$, and
    \item $\forall R'\subset R$, $Sim(Q,R)\ge Sim(Q,R')$
\end{enumerate}
\end{theorem}

\begin{proof}
We prove that for an arbitrary query set $Q$ and an arbitrary group $\mathcal{G}_g$, we can compute a similarity upper bound $UB(Q,\mathcal{G}_g)$ with TGM using Equation (\ref{eq:ub}) such that $\forall S\in\mathcal{G}_g$, $UB(Q,\mathcal{G}_g)\ge Sim(Q,S)$. Let $R=\{t|t\in Q \wedge \exists S\in\mathcal{G}_g,t\in S\}$, then we know $\forall S\in\mathcal{G}_g$, $Q\cap S\subseteq R$. If $Q\cap S= R$, then clearly $Sim(Q,R)\ge Sim(Q,S)$; if $Q\cap S=R'\subset R$, then $Sim(Q,R)\ge Sim(Q,R')\ge Sim(Q,S)$. In either case, $Sim(Q,R)$ upper bounds $Sim(Q,S)$, and thus we can use $Sim(Q,R)$ as $UB(Q,\mathcal{G}_g)$. Since it is possible that $R=S$, in which case $Sim(Q,R)=Sim(Q,S)$, the bound $UB(Q,\mathcal{G}_g)$ is tight, even in multiset settings.
\end{proof}
}

For example, let $Q=\{t_1,t_2,t_3\}$ and $Q\cap S=\{t_1,t_2\}$. Then with Jaccard similarity, the set with the maximal similarity to $Q$ is $\{t_1,t_2\}$ and the upper bound is $\frac{2}{3}$; with cosine similarity, the set with the maximal similarity to $Q$ is also $\{t_1,t_2\}$, but the upper bound is $\frac{2}{\sqrt{3*2}}\approx 0.82$. Note that although most similarity measures satisfy the TGM Applicability Property, some exceptions do exist. One such example is the learned metric \cite{kim2019deep}  which takes two samples (e.g., images) as the input and predicts their similarity.

In what follows, we call the two properties listed in Theorem \ref{theorem:applicable} the {\em TGM Applicability Property}. Note that the token universe $\mathcal{T}$ does not need to be static.  We will discuss how to adapt TGM to deal with cases where $\mathcal{T}$ is dynamically changing in Section~\ref{sec:updates}.


\section{Optimizing Partitioning}
\label{sec:paropt}

{
We analyze how to optimize partitioning to provide higher pruning efficiency. We discuss desired properties of the partitioning, and develop the objective function for the partitioning optimization problem that will guide the development of effective partitioning strategies. To make our formal analysis tractable, we make assumptions regarding the token distribution; nonetheless, as will be demonstrated by our experimental results in Section~\ref{sec:exp}, the optimization objectives and strategies thus developed are also expected to perform well when the assumptions do not hold. 
\subsection{The Case of Uniform Token Distribution}
\label{sec:dpp}
We formally analyze the effect of partitioning on pruning efficiency when the following assumption on token distribution holds. 
\begin{definition} \textbf{Uniform Token Distribution Assumption}.
The probabilities that different tokens belong to an arbitrary set are identical and independent. More specifically, $\forall t_i,t_j\in\mathcal{T}$, $\forall S\in\cS$, $P(t_i\in s)=P(t_j\in S)$, and $P(t_i\in S\vert t_j\in s)=P(t_i\in S\vert t_j\notin S)$.
\end{definition}

For an arbitrary query $Q$, the expected pruning efficiency can be computed as follows:

\begin{footnotesize}
\begin{equation}
\label{eq:range-ub-analysis-delta}
E[PE] = \sum_{g=1}^{n}\vert\mathcal{G}_g\vert (1-UB(Q,\mathcal{G}_g))
\end{equation}
\end{footnotesize}




Given the way the TGM is constructed, we rewrite Equation (\ref{eq:ub}) in the following way to ease subsequent discussion:
\begin{footnotesize}
\begin{equation}
\label{eq:gs}
UB(Q,\mathcal{G}_g)=\frac{\sum_{t\in Q}M[t,g]}{\vert Q\vert}=\frac{\vert GS_g\cap Q\vert}{\vert Q\vert},\ GS_g=\bigcup_{S\in\mathcal{G}_g}S
\end{equation}
\end{footnotesize}

Accordingly, we rewrite Equation (\ref{eq:range-ub-analysis-delta}) as follows:
\begin{footnotesize}
\begin{equation}\label{eq:range-ub-analysis-delta-extended}
E[PE]=\sum_{g=1}^{n}\vert\mathcal{G}_g\vert (1-\frac{\vert GS_g\cap Q\vert}{\vert Q\vert})
\end{equation}
\end{footnotesize}

As we assume $Q$ follows the same distribution as $\cS$, $E[PE]$ over all possible $Q$ can be estimated by the following equation:
\begin{footnotesize}
\begin{equation} \label{eq:expected-pe}
\frac{\sum_{Q\in\cS}\sum_{g=1}^{n}\vert\mathcal{G}_g\vert (1-\frac{\vert GS_g\cap Q\vert}{\vert Q\vert})}{\vert\cS\vert}
\end{equation}
\end{footnotesize}

Since $\vert\cS\vert$ is a constant, we keep the nominator of Equation (\ref{eq:expected-pe}) only, and adjust the order as follows:
\begin{footnotesize}
\begin{equation} \label{eq:expected-pe-omitted}
\sum_{g=1}^{n}\vert\mathcal{G}_g\vert \sum_{Q\in\cS}(1-\frac{\vert GS_g\cap Q\vert}{\vert Q\vert})
\end{equation}
\end{footnotesize}

To ease following analysis, we define term $F$ in Equation (\ref{eq:expected-pe-rewrite}), and claim that maximizing Equation (\ref{eq:expected-pe-omitted}) (and thus maximizing the pruning efficiency) is equivalent to minimizing $F$:
\begin{footnotesize}
\begin{equation} \label{eq:expected-pe-rewrite}
F=\sum_{g=1}^{n}\vert\mathcal{G}_g\vert\sum_{Q\in\cS} \frac{\vert GS_g\cap Q\vert}{\vert Q\vert}
\end{equation}
\end{footnotesize}


We derive several properties regarding the partitioning from Equation (\ref{eq:expected-pe-rewrite}) so as to design practical partitioning algorithms.
\begin{theorem} 
\label{theorem:pe-partitioning}
In a database that satisfies the uniform token distribution assumption, the partitioning that minimizes Equation (\ref{eq:expected-pe-rewrite}) produces groups with equal size (or differ by at most 1).
\end{theorem}
\begin{proof}
We consider the special case where $\mathcal{D}$ is partitioned into two groups $\mathcal{G}_1$ and $\mathcal{G}_2$, and $|\mathcal{G}_1|\le|\mathcal{G}_2|$. The $F$ value of such a partitioning is:
\begin{footnotesize}
\begin{equation}
    F=F_1+F_2= \vert\mathcal{G}_1\vert\sum_{Q\in\cS} \frac{\vert GS_1\cap Q\vert}{\vert Q\vert}+\vert\mathcal{G}_2\vert\sum_{Q\in\cS} \frac{\vert GS_2\cap Q\vert}{\vert Q\vert}
\end{equation}
\end{footnotesize}

Next we move a set $S$ from $\mathcal{G}_1$ to $\mathcal{G}_2$ and prove that such movement increases the $F$ value. We know that if $S$ is moved from $\mathcal{G}_1$ to $\mathcal{G}_2$, $F_1$ would decrease and $F_2$ would increase. And since  $|\mathcal{G}_1|\le|\mathcal{G}_2|$, equivalently we can prove that $\vert\mathcal{G}_i\vert\sum_{Q\in\cS} \frac{\vert GS_i\cap Q\vert}{\vert Q\vert}$ grows super-linearly with respect to $|\mathcal{G}_i|$, or $\sum_{Q\in\cS} \frac{\vert GS_i\cap Q\vert}{\vert Q\vert}$ grows with $|\mathcal{G}_i|$. Given the construction of $GS_i$ in Equation (\ref{eq:gs}), this is evidently true. Therefore, the $F$ value increases after the movement of $S$.

The above discussion can be naturally extended to multi-groups by moving one set from a small group to a large group each time, with the $F$ value increasing and the pruning efficiency decreasing during the process. In conclusion, balanced partitioning results yield the highest pruning efficiency.
\end{proof}
Even though the optimal partitioning is expected to produce groups with almost equal sizes, evidently balance is not the only desired property, according to Equation (\ref{eq:expected-pe-rewrite}). We temporarily omit the $\vert\mathcal{G}_g\vert$ in Equation (\ref{eq:expected-pe-rewrite}) and discuss other properties the partitioning must satisfy in order to provide higher pruning efficiency.
\begin{theorem}
\label{throrem:pe-sim}
In a database that satisfies the uniform token distribution assumption, the partitioning that minimizes the following objective provides the highest pruning efficiency:

\begin{footnotesize}
\begin{equation} \label{eq:pe-sim}
    \sum_{g=1}^n \vert\bigcup_{S\in\mathcal{G}_g}S\vert
\end{equation}
\end{footnotesize}
\end{theorem}
\begin{proof}
Given the assumption that all groups are balanced, minimizing Equation (\ref{eq:expected-pe-rewrite}) is equivalent to minimizing 

\begin{footnotesize}
\begin{equation}
\label{eq:pe-equi}
\sum_{g=1}^n\sum_{Q\in\cS}\frac{\vert GS_g\cap Q\vert}{\vert Q\vert}
\end{equation}
\end{footnotesize}

Since $Q$ follows the uniform token distribution as well, which means that all tokens appear in $Q$ with the same probability, Equation (\ref{eq:pe-equi}) is proportional to the following equation:

\begin{footnotesize}
\begin{equation}
\label{eq:pe-equi-rewrite}
\sum_{g=1}^n\vert GS_g\cap \mathcal{T}\vert=\sum_{g=1}^n\vert GS_g\vert=\sum_{g=1}^n \vert\bigcup_{S\in\mathcal{G}_g}S\vert,
\end{equation}
\end{footnotesize}
where $\mathcal{T}$ denotes the token universe.

Thus, we can maximize PE by minimizing Equation (\ref{eq:pe-sim}).
\end{proof}
}
In summary, we have the following two desired properties regarding the partitioning of database $\cS$. 
\begin{itemize}
    \item Property 1: Groups are balanced;
    \item Property 2: $U=\sum_{g=1}^n \vert\bigcup_{S\in\mathcal{G}_g}S\vert$ is minimized.
\end{itemize}
\vspace{-0.1in}
\subsection{The General Case}
\label{sec:unifiedobj}
The analysis in the preceding section depends on the uniform token distribution assumption. In real-life datasets, this assumption does not hold. However, following the same methodology to derive a formal treatment of an arbitrary set/token distribution would be challenging as a realistic mathematical model of arbitrary set/token distributions would be hard to justify. Although the two properties identified above may not be true for optimal partitioning in the general case, we draw inspirations from them and propose a heuristic objective function that strives to maximize PE.

In essence Property 2 directs that the more similar (in terms of token composition) the sets are within a group, the better. We thus design a {\em general partitioning objective ($GPO$)} we wish to minimize reflecting this property:

\begin{footnotesize}
\begin{equation}
    \label{eq:unified-obj}
    GPO=\sum_{g=1}^{n}\sum_{S_x\in\mathcal{G}_g}\sum_{S_y\in\mathcal{G}_g}(1-Sim(S_x,S_y)),
\end{equation}
\end{footnotesize}
where $Sim(*)$ can be any measures discussed in Section \ref{sec:applicability}.

Intuitively, $GPO$ aims to minimize the sum of the intra-group pair-wise distances, where {\em distance} is defined as $1-Sim(*)$. This is similar to Property 2. As an example, consider two groups $\mathcal{G}_i$ and $\mathcal{G}_j$ with $\vert\mathcal{G}_i\vert=\vert\mathcal{G}_j\vert$. Assume that  $Sim(*)$ is Jaccard similarity, and we are to place a new set $S$ into one of the two groups. Then, if $\sum_{S_i\in\mathcal{G}_i}(1-Sim(S,S_i))<\sum_{S_j\in\mathcal{G}_j}(1-Sim(S,S_j))$, that would mean $S$ shares more common tokens with sets in group $\mathcal{G}_i$, and thus inserting $S$ into group $\mathcal{G}_i$ helps to minimize $U$.

However, considering only Property 2 results in highly skewed partitioning results, as placing all sets in the same group provides the minimal $U$ (which equals to $\vert \mathcal{T}\vert$). Evidently, Property 1 is used to prevent such skewed partitioning in the uniform case. Luckily, $GPO$ enjoys a similar functionality: placing all sets in the same group provides $GPO=\sum_{S_x,S_y\in\cS}(1-Sim(S_x,S_y))$, which is the maximal possible $GPO$, and thus such a partitioning is never the optimal in terms of $GPO$. Thus, the design of $GPO$ implicitly incorporates both Property 1 and Property 2.


In order to better appreciate the distinctive value of the proposed partitioning objective, we compare $GPO$ with $k$-medians, perhaps the most popular clustering technique, and show by an example how optimizing $GPO$ leads to better results. We use a database of 21 sets, and the partitioning results based on different clustering objectives are given in Figure \ref{fig:illus-compare}. Each set is represented by a point in the plot, and to better visualize the results we replace $(1-Sim(*))$ with Euclidean distance.
\begin{figure}
\centering
    \subfloat[Groups produced by $GPO$]{\includegraphics[width=2in]{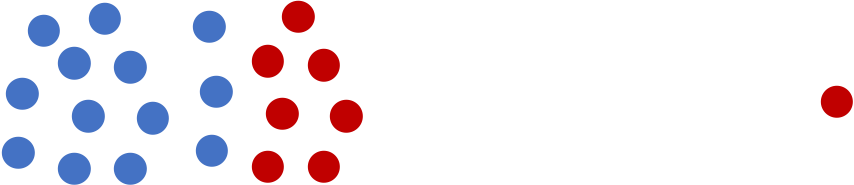}}
    \vspace{-0.1in}
    \subfloat[Groups produced by $k$-medians]{\includegraphics[width=2in]{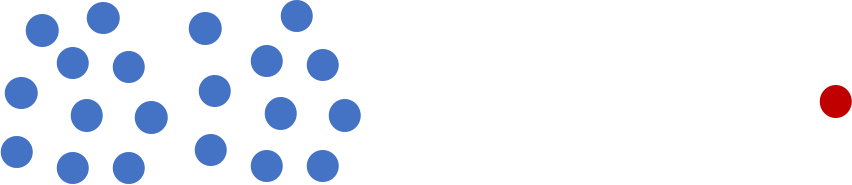}}
\caption{Comparison of different partitioning results }
\label{fig:illus-compare}
\vspace{-0.05in}
\end{figure}

Assume that the query is to identify the nearest neighbors of all 21 points. According to the search strategy of \les\ given in Section \ref{sec:index}, all points in the same group are candidates of each other. Therefore, with the clustering results given in Figure \ref{fig:illus-compare}(b), the total number of distance calculations is $20*20+1*1=401$, while with the partitioning results in Figure \ref{fig:illus-compare}(a) the number is $13*13+8*8=233$. Clearly, the results based on Equation (\ref{eq:unified-obj}) have better pruning efficiency. 

\begin{theorem}
\label{theorem:np}
Given a database of sets $\cS$ minimizing $GPO$ on $\cS$ is NP-complete.
\end{theorem}
{
\begin{proof}
We give a brief proof by showing that minimizing $GPO$ is essentially a 0-1 integer linear programming problem, which has been shown to be NP-complete \cite{cook1971complexity}. 
More specifically, minimizing $GPO$ is equivalent to solving the following optimization problem:

\begin{equation}
\label{eq:np}
\begin{aligned}
\text{maximize } \quad & \mathbf{e}_{|\mathcal{D}|}\cdot[\mathbf{A}\cdot \mathbf{A}^\intercal\odot \mathbf{D}]\cdot \mathbf{e}_{|\mathcal{D}|}^\intercal \\
\text{subject to }\quad & \mathbf{e}_n\cdot\mathbf{A}^\intercal=\mathbf{e}_{|\mathcal{D}|}\\
\end{aligned}   
\end{equation}
where $\mathbf{A}$ is a $|\mathcal{D}|\times n$ matrix and $\mathbf{A}[x,g]=1$ if set $S_x$ belongs to group $\mathcal{G}_g$ and $\mathbf{A}[x,g]=0$ otherwise, and $\mathbf{D}$ of size $|\mathcal{D}|\times |\mathcal{D}|$ denotes the distance matrix where $\mathbf{D}[x,y]=1-Sim(S_x,S_y)$, and $\mathbf{e}_i$ is a row vector of length $i$ filled with ones. The goal is to find the $\mathbf{A}$ which satisfies the constraint and maximizes the objective.

The intuition behind Equation (\ref{eq:np}) can be described as follows: $\mathbf{A}\cdot\mathbf{A}^\intercal$ is a $|\mathcal{D}|\times |\mathcal{D}|$ matrix such that the value at position $[x,y]$ is 1 if $S_x$ and $S_y$ belong to the same group, and 0 otherwise. The element-wise product between $\mathbf{A}\cdot\mathbf{A}^\intercal$ and $\mathbf{D}$ masks out those pair-wise distances between sets belonging to different groups, and $ \mathbf{e}_{|\mathcal{D}|}\cdot[\mathbf{A}\cdot \mathbf{A}^\intercal\odot \mathbf{D}]\cdot \mathbf{e}_{|\mathcal{D}|}^\intercal$ sums the remaining distances, which is the same objective as $GPO$. The constraint $\mathbf{e}_n\cdot\mathbf{A}^\intercal=\mathbf{e}_{|\mathcal{D}|}$ guarantees that each set belongs to one and only one group. Therefore, minimizing $GPO$ is equivalent to solving Equation (\ref{eq:np}), which completes the proof.
\end{proof}
}


\vspace{-0.05in}
\subsection{Algorithmic Approaches}
\label{sec:heuristic}
In this section we propose several algorithmic approaches based on existing applicable clustering methods, which are expected to yield groups with low $GPO$ values. More specifically, we design a graph cut-based approach (PAR-G), a centroid-based approach (PAR-C), and a hierarchical approach (PAR-H).

{

\subsubsection{Graph cut-based method (PAR-G)}
\label{sec:par-g}
When $k$ or $\delta$ is fixed, it is possible to build an index structure specifically optimized for the workload.  
Dong et al. \cite{Dong2020Learning} propose a graph cut-based solution for (approximate) nearest neighbor search in $\mathbb{R}^d$ space by linking each point to its neighbors and partitioning the resulting graph into balanced subgraphs with the number of edges crossing different subgraphs minimized. Such a partitioning is shown to yield high pruning efficiency. Inspired by their approach, we design PAR-G, which takes $k$ or $\delta$ as one of its inputs, as follows:
\begin{enumerate}
    \item {\bf Similarity graph construction.} For a given $k$ in $k$NN query, construct the similarity graph, $G_{\mathcal{D}}$, of $\mathcal{D}$, such that $\forall S_x\in\mathcal{D}$, there exists a corresponding vertex $V_x$ in $G_{\mathcal{D}}$, and $\forall S_y\in\mathcal{D}$, if $S_y$ is one of the $k$ nearest neighbors of $S_x$, there is an edge between $V_x$ and $V_y$ in $G_{\mathcal{D}}$. For a given $\delta$ in range query, there is an edge between $V_x$ and $V_y$ if $Sim(S_x,S_y)\ge\delta$. 
    \item {\bf Graph cut.} Partition $G_{\mathcal{D}}$ into $n$ balanced subgraphs while minimizing the number of edges crossing different subgraphs. This can be done with existing graph partitioners \cite{karypis1999multilevel,DBLP:conf/wea/GottesburenHSW20}.
\end{enumerate}

}
\subsubsection{Centroid-based method (PAR-C)}
\label{sec:par-c}
Centroid-based methods \cite{hartigan1979algorithm} are iterative algorithms which at each iteration relocate an elements into a different cluster if such relocation improves the overall objective function. For our case, let $\phi(\mathcal{G})=\sum_{S_x,S_y\in\mathcal{G}}(1-Sim(S_x,S_y))$ be the sum of all pair-wise distances\footnote{Note that repetitively calculating $\phi(\mathcal{G})$ during the partitioning process is computational prohibitive, and thus we approximate $\phi(\mathcal{G})$ with randomly selected sets in $\mathcal{G}$ in the experiment (Section \ref{sec:exp-alg}).} in group $\mathcal{G}$, $S\in\mathcal{G}_i$ an arbitrary set, and $\Delta(S,\mathcal{G}_i,\mathcal{G}_{j})=\phi(\mathcal{G}_i\setminus S)+\phi(\mathcal{G}_j\cup S)-\phi(\mathcal{G}_i)-\phi(\mathcal{G}_j)$ the decrease of $GPO$ after moving $S$ from $\mathcal{G}_i$ to $\mathcal{G}_j$ ($i,j\in[1,n]$). To be more specific, our method works as follows:
\begin{enumerate}
    \item {\bf Initialization.} Randomly partition $\cS$ into $n$ groups;
    \item {\bf Relocation.} For each $S\in\cS$, suppose $S\in\mathcal{G}_i$. Find group $\mathcal{G}^*_j$ such that $\Delta(S,\mathcal{G}_i,\mathcal{G}^*_j)=\max_{S,\mathcal{G}_i,\mathcal{G}_j}\Delta(S,\mathcal{G}_i,\mathcal{G}_j)$ (denoted as ``the best group''). If $\Delta(S,\mathcal{G}_i,\mathcal{G}_j^*)>0$, relocate $S$ from $\mathcal{G}_i$ to $\mathcal{G}_j^*$. Repeat this step until no sets are relocated in an iteration.
    \item {\bf Simplification.} Considering the data size we deal with (see Section \ref{sec:exp}), finding ``the best group'' at each iteration would be too expensive. Therefore, we adopt the ``first-improvement'' variant \cite{telgarsky2010hartigan} of the algorithm, i.e., pick the first group $\mathcal{G}_j$ with $\Delta(S,\mathcal{G}_i,\mathcal{G}_j)>0$ rather than the best group.
\end{enumerate}


\subsubsection{Divisive clustering method (PAR-D)} Divisive clustering methods \cite{kaufman2009finding} start from the single cluster containing all elements and repeatedly split clusters until a desired number of clusters is reached. We reuse $\phi(\mathcal{G})$ introduced in Section \ref{sec:par-c} and use $idv\_d(S)$ to denote the sum of distances between $S$ and all other sets in the same group as $S$. PAR-D works as follows:

\begin{enumerate}
    \item {\bf Initialization.} Take $\cS$ as the initial group.
    \item {\bf Splitting.} Find group $\mathcal{G}^*=\arg\max_{\mathcal{G}_i\in\{\mathcal{G}_1,\mathcal{G}_2,\cdots\}}\phi(\mathcal{G}_i)$, where $\{\mathcal{G}_1,\mathcal{G}_2,\cdots\}$ denotes all current groups. Find set\\ $S^*=\arg\max_{S\in\mathcal{G}^*}idv\_d(S)$. Create a new group $\mathcal{G}^{new}=\{S^*\}$. For all other sets $S'\in\mathcal{G}^*$, move $S'$ to $\mathcal{G}^{new}$ if such movement reduces
   the overall $GPO$. Repeat this step until there are $n$ groups.
   \item {\bf Simplification.} Considering the data size we deal with, instead of finding $S^*$, we choose a random set in $\mathcal{G}^*$ to initialize $\mathcal{G}^{new}$, which is commonly adopted for group splitting \cite{guttman1984r}.
\end{enumerate}

{
\subsubsection{Agglomerative clustering method (PAR-A)} Agglomerative clustering \cite{rokach2005clustering} works in a bottom-up fashion by initially treating each element as a cluster and repeatedly merging clusters until a desired number of clusters is reached. We reuse $\phi(\mathcal{G})$ introduced in Section \ref{sec:par-c} to denote the sum of all pair-wise distances in group $\mathcal{G}$. PAR-A works as follows:

\begin{enumerate}
    \item {\bf Initialization.} Create group $\mathcal{G}_i=\{S_i\}$ for each $S_i\in\mathcal{D}$.
    \item {\bf Merging.} Find groups $\mathcal{G}^*_1,\mathcal{G}^*_2=\arg\min_{\mathcal{G}_i,\mathcal{G}_j\in\{\mathcal{G}_1,\mathcal{G}_2,\cdots\}}\\ \phi(\mathcal{G}_i\cup\mathcal{G}_j)$, where $\{\mathcal{G}_1,\mathcal{G}_2,\cdots\}$ denotes all current groups and $i\neq j$. Create a new group $\mathcal{G}^{new}=\mathcal{G}^*_1\cup\mathcal{G}^*_2$ and remove groups $\mathcal{G}^*_1$ and $\mathcal{G}^*_2$. Repeat this step until there are $n$ groups.
   \item {\bf Simplification.} Considering the data size we deal with, we adopt the heuristic that merging smaller groups (groups with smaller number of sets) usually results in smaller values of $\phi(\mathcal{G}_i\cup\mathcal{G}_j)$ and restrict that $\mathcal{G}^*_1$ is the smallest group (breaking ties randomly), and thus only $\mathcal{G}^*_2$ needs to be identified in each iteration.
\end{enumerate}

}

As we will demonstrate in our experimental study in Section 7, these heuristic approaches do not provide satisfactory performance. The structure of the $GPO$ problem objective does not resemble those targeted by well-studied clustering algorithms. In the next section, we explore the use of ML to perform such a partitioning.

\section{L2P: learn to partition sets into groups}
\label{sec:l2p}
{As pointed out by Bengio et al. \cite{bengio2018machine}, a machine learning approach to combinatorial optimization problems with well-defined objective functions, such as the Travelling Salesman Problem, has proven to be more promising than classical optimization methods with hand-wired rules in many scenarios, for the reason that it adapts solutions to the data and thus can uncover patterns in the specific problem instance as opposed to solving a general problem for every instance. It is widely agreed \cite{vinyals2015pointer,bahdanau2014neural} that ML-based methods are especially valuable in cases where expert knowledge of the problem domain may not be sufficient and some algorithmic decisions may not give satisfactory results. Our goal in this section, therefore, is to develop a machine learning method to optimize $GPO$.}
\subsection{Siamese Networks}
\label{sec:siamese}
Considering that the goal of optimizing $GPO$ is to maximize the overall intra-group similarity, we adopt Siamese networks \cite{bromley1994signature,fan2018adversarial} to solve the partitioning task. Siamese networks have been successfully utilized in deep metric learning tasks \cite{mueller2016siamese,neculoiu2016learning} in computer vision, capturing both intra-class similarity and inter-class discrimination in many challenging tasks including face recognition.


We design a Siamese network as shown in Figure \ref{fig:siamesenet} to learn the optimal partitioning. It consists of a pair of twin networks sharing the same set of parameters working in tandem on two inputs and generate two comparable outputs.   
\begin{figure}[h]
\vspace{-0.1in}
    \centering
    \includegraphics[width=2.2in]{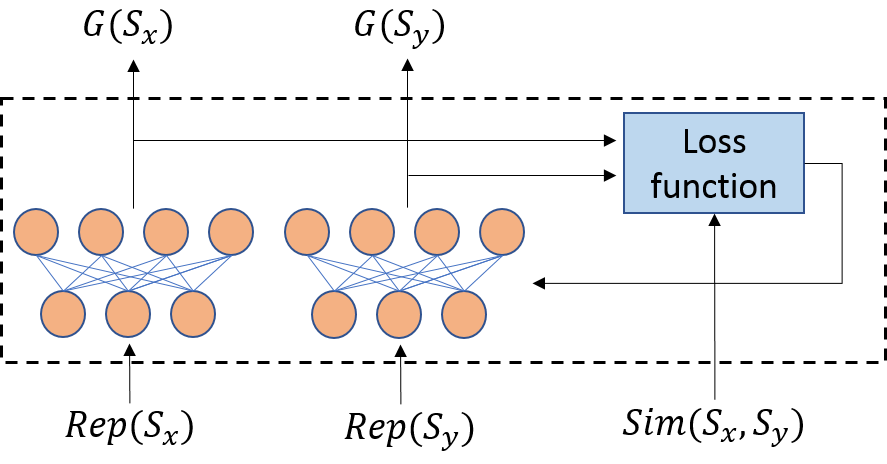}
    \caption{Siamese network}
    \label{fig:siamesenet}
\vspace{-0.1in}
\end{figure}

We use $Rep(S_x)$ and $Rep(S_y)$ to denote the vector representations of two sets $S_x$ and $S_y$, a pair of inputs to the twin networks, and use $G(S_x)$ and $G(S_y)$ to represent their respective group assignment indicted by the outputs of the twin networks respectively. Following Equation (\ref{eq:unified-obj}) we define the loss function of the Siamese network as follows:

\begin{footnotesize}
\begin{equation}
\label{eq:siameseloss}
    loss(S_x,S_y)=
    \begin{cases}
      (1-Sim(S_x,S_y)), & \text{if}\ G(S_x)=G(S_y)\\
      0, & \text{otherwise}
    \end{cases}
\end{equation}
\end{footnotesize}

{Equation (\ref{eq:siameseloss}) minimizes the intra-group dissimilarities by summing $(1-sim(S_x,S_y))$ as the losses, and penalizes imbalanced groups by counting pairwise dissimilarities only between sets in the same group. We use an example to illustrate how Equation (\ref{eq:siameseloss}) penalizes imbalanced partitioning. Suppose there are $N$ sets and dissimilarity between any pair of sets is the same at $d$. The task is to partition these sets into 2 groups, containing $N_1$ and $N_2$ sets respectively ($N_1+N_2=N$). Then the overall loss is $\frac{N_1(N_1-1)d}{2}+\frac{N_2(N_2-1)d}{2}=\frac{d}{2}[N_1(N_1-1)+N_2(N_2-1)]=\frac{d}{2}(N_1^2+N_2^2-N)\ge\frac{d}{2}(\frac{(N_1+N_2)^2}{2}-N)=\frac{d}{2}(\frac{N^2}{2}-N)$, and the bound is tight when $N_1=N_2$. Therefore, Equation (\ref{eq:siameseloss}) favors balanced partitioning.}

By training the Siamese network with sufficient samples drawn from $\mathcal{D}$, theoretically we can minimize the overall distances between all pairs of sets which belong to the same group, and thus the Siamese network is expected to give the partitioning result in which $GPO$ is minimized. Practically, however, we expect to achieve near-optimal partitioning only as the network is essentially performing local search. 
\vspace{-0.1in}
\subsection{Framework}
\label{sec:framework}
Although using Siamese networks to solve the optimization problem is a promising approach, training such networks turns out to be difficult for the following reasons:
\begin{enumerate}
    \item When dealing with real world data, we may need to partition sets into thousands of groups. Therefore, for an input set $S_x$, the network needs to make prediction on which group $S_x$ belongs to, among a collection of thousands of groups. It is well known \cite{gupta2014training} that training networks to tackle prediction problems involving thousands or more labels is challenging.
    \item What makes this task even more difficult is that unlike a classification problem, the label for each input (i.e., the optimal group) in this optimization problem is unknown, i.e., there is no ground truth regarding the labels/groups available. The only information we have is the loss if the two input sets are assigned into the same group. This makes the problem even more challenging than typical classification problems.
\end{enumerate}

The inherent difficulty of utilizing Siamese networks for this problem is the dimensionality (i.e., degrees of freedom) of the output space. In response to this challenge, we propose a learning framework consisting of a cascade of Siamese models, which partitions the database in a hierarchical fashion. Each Siamese network in the framework is responsible for partitioning a group of sets into two sub-groups. The framework is  illustrated in Figure~\ref{fig:cascade}.

\begin{figure}[h]
    \centering
    \includegraphics[height=1.3in]{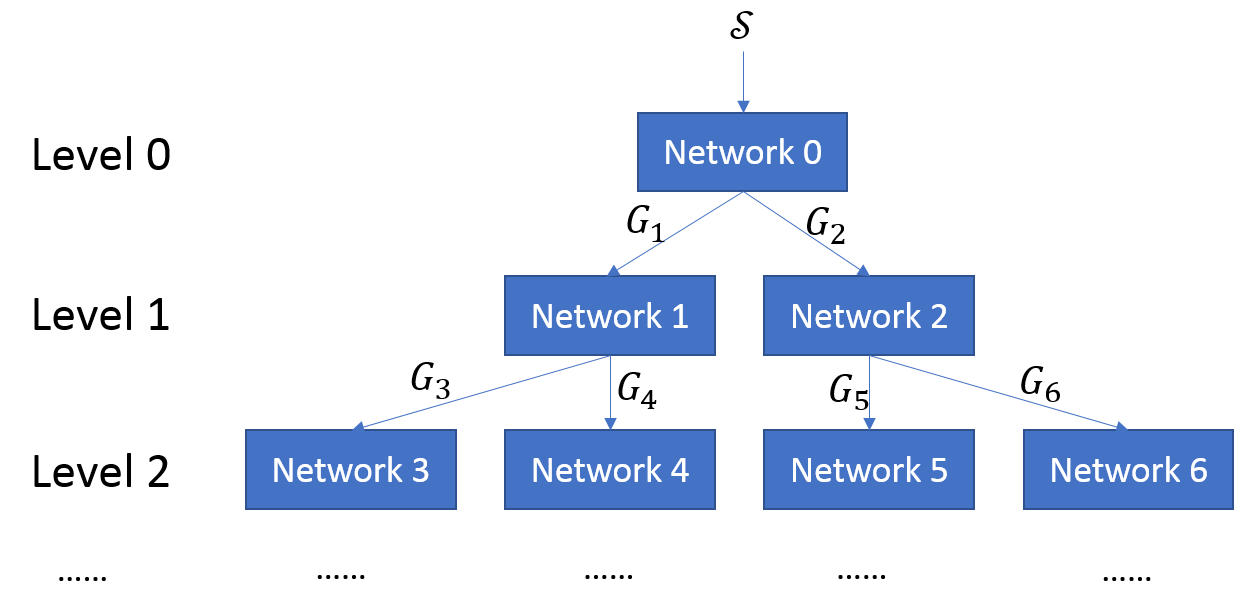}
    \vspace{-0.05in}
    \caption{Cascade framework}
    \label{fig:cascade}
    \vspace{-0.1in}
\end{figure}

At Level 0 of the framework, we train a Siamese network which takes each set in the entire database $\mathcal{D}$ and assigns it into one of two groups, $\mathcal{G}_1$ and $\mathcal{G}_2$, based on the loss function given in Equation~(\ref{eq:siameseloss}). Then at Level 1, we train two Siamese networks working on $\mathcal{G}_1$ and $\mathcal{G}_2$ respectively in the same fashion. Thus, they partition the entire database into four groups. We continue adding more levels to the cascade framework until all groups are small enough or a pre-defined threshold on the number of levels is reached. Since each model in the cascade is specialized to partition a group of sets into only two sub-groups the resulting classification problem can be solved effectively.


The architecture of the cascade models motivates the use of a hierarchical indexing structure, which we call Hierarchical TGM (or HTGM). More specifically, assuming the level of the cascade framework is $l$, and $0\le i<j<l$, we construct TGM$_i$ and TGM$_j$ based on the partitioning results at level $i$ and level $j$ respectively. Suppose a group at level $i$, say $\mathcal{G}_g$, is partitioned into several sub-groups at level $j$, say $\mathcal{SG}_1,\cdots,\mathcal{SG}_m$. If for a query $Q$, group $\mathcal{G}_g$ can be pruned by checking TGM$_i$, then all verification operations involving groups $\mathcal{SG}_1,\cdots,\mathcal{SG}_m$ can be eliminated. The construction can be easily generalized to HTGM with $h$ ($h>1$) levels. 
\vspace{-0.05in}
\subsection{PTR: a Set Representation Method}
\label{sec:setrep}
A Siamese network accepts vectors as input and thus the sets in $\mathcal{D}$ cannot be directly fed into the network. As a result we have to build a vector representation for each set. {Considering the time and space complexity of existing embedding methods such as Principal Component Analysis (PCA) or Multidimensional Scaling (MDS), they can hardly be applied to the target setting introduced in Section \ref{sec:exp} where millions or billions of sets are involved (see comparison regarding embedding cost in Section \ref{sec:exp-rep}). Besides, different from the objectives of these general-purpose embedding methods such as maximizing variance or preserving distance,} {our concern is to make sets containing different tokens more separable to benefit the training of the Siamese network. Intuitively, the representations that ease the training of the models are expected to bear the following property:
\begin{definition} \textbf{Set Separation-Friendly Property}.
$\forall t\in\mathcal{T}$, let $\mathcal{G}_t$ be the collection of sets containing $t$, and  $\mathcal{G}_{\neg t}$ be the collection sets not containing $t$, then $\mathcal{G}_t$ and  $\mathcal{G}_{\neg t}$ should to be easily separable in the representation space.
\end{definition}

Next we discuss how to construct such representations.
}
As the first step, we organize all tokens with a balanced binary tree such that tokens appear in leaf nodes and each leaf contains only one token. The height of the tree is thus $h=\lceil \log_2{\vert\mathcal{T}\vert}\rceil$. We mark the edge from a node to its left child with 1 and the edge to its right child with 0. An example of such a tree is depicted in Figure \ref{fig:tree}.
\begin{figure}[h]
\vspace{-0.1in}
    \centering
    \includegraphics[width=1.2in]{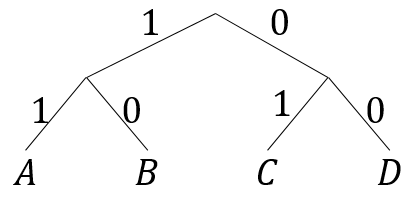}
    \caption{Tokens organized with a balanced tree}
    \label{fig:tree}
\vspace{-0.05in}
\end{figure}

We use $path_t$ to denote the path from the root to an arbitrary token $t$. Since each leaf contains only one token, no two tokens share the same path. We build a path table (PT) of all tokens defined as follows:

\begin{footnotesize}
\begin{equation}
    \label{eq:idt}
     \mbox{PT}[t,i]=
    \begin{cases}
      path_t[i], & \text{if}\ i\in[1,h]\\
      1-path_t[i], & \text{if}\ i\in[h+1,2h]
    \end{cases}
\end{equation}
\end{footnotesize}

An example of PT is provided in Table \ref{tab:pt}.

\begin{table}[h]
\caption{An example of path table (PT)}
\centering
\begin{tabular}{|c|c|c|c|c|}
\hline
Position & 1 & 2 & 3 & 4 \\ \hhline{|=|=|=|=|=|}
A & 1 & 1 & 0 & 0 \\ \hline
B & 1 & 0 & 0 & 1 \\ \hline
C & 0 & 1 & 1 & 0 \\ \hline
D & 0 & 0 & 1 & 1 \\ \hline
\end{tabular}
\label{tab:pt}
\end{table}


We propose a method called PTR (\underline{P}ath \underline{T}able \underline{R}epresentation) to build a representation for a given set $S$ as follows:
\begin{footnotesize}
\begin{equation}
    \label{eq:representation}
    Rep(S)[i]=\sum_{t\in S}PT[t,i],\ i\in[1,2h]
\end{equation}
\end{footnotesize}

In the above example, the representation of $\{A,B,C\}$ is $[2,2,1,1]$ and the representation of $\{B,D\}$ is $[1,0,1,2]$. { The second half of the path table (Positions 3 and 4) helps to reduce the chance that different sets have common representations. For example, if only the first half is used, then the representations of $\{A\}$, $\{B,C\}$, $\{A,D\}$, $\{B,C,D\}$ would all be $[1,1]$. We compare the set representations constructed on the full vs. half path tables in Section \ref{sec:exp-rep}.}

{PTR also naturally differentiates multisets containing the same collection of tokens but with different number of occurrences. For example, $Rep(\{A\})=[1,1,0,0]$ while $Rep(\{A,A\})=[2,2,0,0]$.}


The basic idea of the representation is to map the sets into a new space, in a way that determining collections of sets containing specific tokens can be easily performed. More specifically, set $S$ is placed in the representation space based on the presence or absence of all tokens in $S$, and consequently, given a collection of tokens $\mathcal{T}_c$, we can quickly locate all sets containing $\mathcal{T}_c$. This evidently yields the set separation-friendly property. To better illustrate this, we reuse the path table in Table \ref{tab:pt} and show that sets containing $B$ can be separated from other sets. For better visualization, we project the representation space onto the first two dimensions (Positions 1 and 2), and keep tokens $B$, $C$, and $D$ only  in Figure~\ref{fig:hyperplane}. Clearly all sets containing token $B$ fall into the striped area, defined by the axis aligned hyper-plane in the representation space passing from point $(1,0)$ (corresponding to $\{B\}$). Similarly all sets containing both $B$ and $C$ are located at the intersection of the axis aligned hyper-planes passing from $(1,0)$ and $(0,1)$ (corresponding to $\{C\}$). That way separating sets in the representation space based on token membership is conducted by determining hyper-plane intersections. We will demonstrate that such a representation is easier to learn and yields effective partitions in Section \ref{sec:exp}.
\begin{figure}
    \centering
    \includegraphics[height=1.1in]{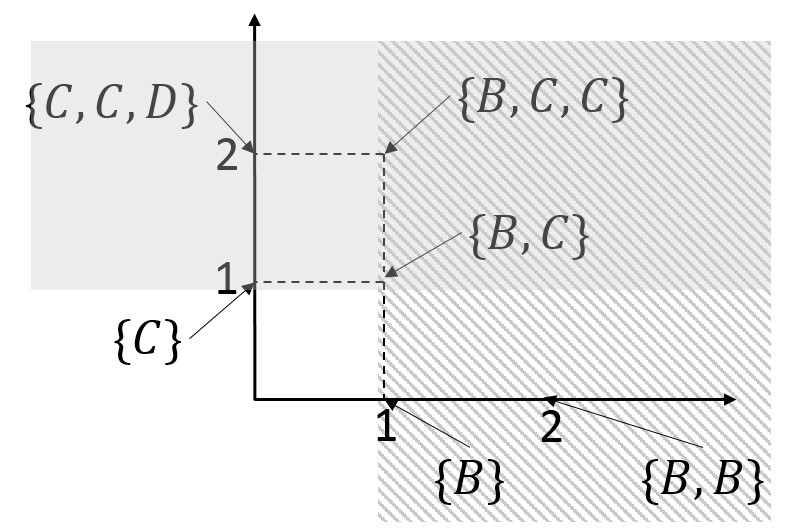}
    \caption{Separating sets}
    \label{fig:hyperplane}
\vspace{-0.1in}
\end{figure}

\vspace{-0.05in}
\section{Dealing with Updates}
\label{sec:updates}
Our discussions so far have assumed that the database $\cS$ and the token universe are fixed. In some cases, however, new sets may be added to the database after the index is built, and previously unseen tokens may appear. We therefore study how updates can be handled, with a focus on TGM, as HTGM can be updated level by level in the same way.

We first discuss the case where new sets are added but the token universe remains the same. Given a new set $S$, we add $S$ into the group $\mathcal{G}_g$ if the similarity upper bound between $\mathcal{G}_g$ and $S$ is the highest among all groups. When there exist multiple groups with the same highest $UB$, we insert $S$ into the group with the minimal number of sets, in line with the optimization target discussed in Section \ref{sec:paropt}. After insertion, we update the TGM accordingly, i.e., for all tokens $t\in S$, we set $M[g,t]=1$.

We now demonstrate how our approach naturally handles previously unseen tokens. This is an important feature of our solution as it is the first to deal with dynamic tokens. All previous attempts to a solution of this problem assumed a fixed token universe \cite{zhang2017efficient,DBLP:journals/tkde/ZhangWWX20}. Let $S$ be a set containing one or more new tokens. We insert $S$ into the database in the following two steps:
\begin{enumerate}
    \item Let $PS= S\cap\mathcal{T}$ be all tokens in $S$ that have been seen previously. We find the group with the highest similarity upper bound to $PS$, denoted by $\mathcal{G}_g$. If $PS=\emptyset$, then $\mathcal{G}_g$ is simply the group with the minimal number of sets. $S$ is inserted to $\mathcal{G}_g$.
    \item For any token $t_{new}\in S\setminus PS$, add a row in $M$ corresponding to $t_{new}$. For all tokens $t\in S$, set $M[g,t]=1$.
\end{enumerate}
Although the partitioning in Section \ref{sec:paropt} is optimized based on existing sets and tokens, inserting new sets and tokens will not severely impact the performance of the approach, as we demonstrate in Section \ref{sec:update}. 
\vspace{-0.05in}
\section{Experiments}\label{sec:exp}
In this section, we present a thorough experimental evaluation of our approach varying parameters of interest, comparing \les\ and its important components, L2P and PTR, with competing methods.
\vspace{-0.05in}
\subsection{Settings}
\label{sec:exp-settings}
{\bf Environment.} We run the experiments on a machine with an Intel(R) Core i7-6700 CPU, 16GB memory and a 500GB, 5400 RPM HDD (roughly 80MB/s data read rate). {We use HDD for fair comparison as other disk-based methods require no random access of the data (see Section~\ref{sec:sssbaselines}). However one could expect better performance of \les\ when running on SDD as it incurs random access of the data by skipping some groups, especially when the number of groups is large.} 

{\bf Implementation.} L2P is implemented with PyTorch, embedding methods in Section \ref{sec:exp-rep} and  partitioning methods in Section \ref{sec:exp-alg} are implemented with Python, and TGM and the set similarity search baselines are implemented in C++ and compiled using GCC 9.3 with -O3 flag. TGM is compressed by Roaring \cite{lemire2018roaring}, a well-performed bitmap compression technique.

{\bf Datasets.} KOSARAK \cite{KOSARAK-online}, LIVEJ \cite{mislove-2007-socialnetworks}, DBLP \cite{DBLP-online}, and AOL \cite{pass2006picture} are three popular datasets used for set similarity search problems and we adapt them for this reason. We also include a social network dataset from Friendster\cite{yang2015defining} (denoted by FS), where each user is treated as a set with his/her friends being the tokens; and a dataset from PubMed Central journal literature \cite{pmcref} (denoted by PMC), where each sentence is treated as a set with the words being the tokens\footnote{with basic data cleaning operations such as stop-words removal.}. Table \ref{tab:dataset} presents a summary of the statistics on these datasets. Considering the size of FS and PMC, we utilize them for disk-based evaluation in Section \ref{sec:sssbaselines} to examine the scalability of \les.
\begin{table}[h]
\vspace{-0.1in}
\caption{Dataset statistics}
\begin{tabular}{|c|r|r|c|r|r|}
\hline
\multirow{2}{*}{Dataset} & \multicolumn{1}{c|}{\multirow{2}{*}{$\vert\cS\vert$}} & \multicolumn{3}{c|}{Set size}                             & \multicolumn{1}{c|}{\multirow{2}{*}{$\vert\mathcal{T}\vert$}} \\ \cline{3-5}
                         & \multicolumn{1}{c|}{}                   & \multicolumn{1}{c|}{Max} & Min & \multicolumn{1}{c|}{Avg} & \multicolumn{1}{c|}{}                   \\ \hhline{|=|=|=|=|=|=|}
KOSARAK                  & 990,002                                 & 2,498                    & 1   & 8.1                      & 41,270                                  \\ \hline

LIVEJ                    & 3,201,202                               & 300                      & 1   & 35.1                     & 7,489,073                               \\ \hline
DBLP                     & 5,875,251                               & 462                      & 2   & 8.7                    & 3,720,067                               \\ \hline
AOL                      & 10,154,742                              & 245                      & 1   & 3.0                      & 3,849,555                               \\ \hline
FS                       & 65,608,366                              & 3,615                    & 1   & 27.5                     & 65,608,366                              \\ \hline
PMC                      & 787,220,474                             & 2,597                    & 1   & 8.8                      & 22,923,401                              \\ \hline
\end{tabular}
\label{tab:dataset}
\end{table}
\vspace{-0.1in}

{\bf Evaluation.} Following previous studies \cite{deng2018overlap,mann2016empirical,DBLP:journals/tkde/ZhangWWX20}, we adopt Jaccard similarity as the metric in our experimental evaluation. We stress however that any similarity measures satisfying the TGM applicability property introduced in Section \ref{sec:applicability} can be adopted in our framework with highly similar results as those reported below. {For each experiment, we randomly select $10$K sets in the corresponding dataset as the queries and report the average search time.} Unless otherwise specified, the indexing structure (TGM) and the data are memory-resident. We conduct disk-based evaluation in Section \ref{sec:sssbaselines}. We compare TGM with HTGM in Section \ref{sec:tgm-htgm}. We select $n$ (number of groups) for each dataset which results in the shortest query latency. The influence of $n$ is studied in Section \ref{sec:sensitivitytonandk}.


{\bf Network and Loss Function.} We consider Multi-Layer Perceptron (two hidden layers, each consisting of eight neurons) and Sigmoid activation function for L2P training in the experiment and leave the investigation of other networks as future work. Clearly the network has one neuron at the output layer. Let $O_x$ be the output on input set $S_x$. If $O_x<0.5$, then $S_x$ belongs to the first group; if $O_x\ge 0.5$, then $S_x$ belongs to the second group.

The loss function given in Equation (\ref{eq:siameseloss}) clearly describes the learning objective. However, it is difficult to train the network with that loss function as its gradient is 0 for most outputs (to be exact, the gradient is $(1-sim(S_x,S_y))$ when $O_x=O_y=0.5$, and is 0 elsewhere). For efficient training, we use the following surrogate loss function, which leads to the same global optimum as Equation (\ref{eq:siameseloss}) while introducing non-zero gradients:

\begin{footnotesize}
\begin{equation}
\label{loss:final}
    loss'(S_x,S_y)=
    \begin{cases}
      W(O_x,O_y)(1-Sim(S_x,S_y)), & \text{if}\ V(O_x,O_y);\\
      0, & \text{otherwise;}
    \end{cases}
\end{equation}
\end{footnotesize}
where $W(O_x,O_y) = (0.5-\vert O_x-O_y\vert)$, and  $V(O_x,O_y) = [(O_x\ge 0.5\wedge O_y\ge 0.5)\vee(O_x< 0.5\wedge O_y< 0.5)]$.

{
{\bf Initialization.}  Models at the first few levels of the Cascade framework deal with a large number of sets and incur long training time. To improve the training efficiency, we first sort all sets based on the minimal token contained in each set, and then partition all sets into 128 groups such that each group contains consecutive $|\mathcal{D}|/128$ sets, inspired by the idea of imposing sequential constraint to clustering tasks \cite{szkaliczki2016clustering}.  Since we always build TGM on the partitioning results at level 10 or higher which may contain thousands of groups, such initialization has minor impact on the performance but greatly reduces the training time. Note that initialization is not performed for the sampled dataset used in Section \ref{sec:exp-rep} due to its small size.
}

{\bf Training.} For the Siamese network partitioning an arbitrary group, we randomly select 40,000 pairs of sets in the group to generate training samples, relatively small compared to the data size. It is observed that further increasing the number of training samples do not significantly improve the pruning efficiency of the partitioning results. We stop partitioning a group if it contains less than 50 sets, and thus the number of groups at level $i$ may be less than $2^i$. The batch size is set to 256, the number of epochs is set to 3 (except for Section \ref{sec:train-cost} which reports the learning curves), and Adam is used as the optimizer. {The same sampling-and-training procedure is repeated for each model in the cascade framework, starting from level 0.}


{
\subsection{Model Convergence and Training Cost}
\label{sec:train-cost}
In this section we report the learning curves and the training costs. We observe that different models in the cascade framework introduced in Section \ref{sec:framework} yield similar learning curves, and thus we present the training loss of a random model at level 0 for each dataset (note that there are 128 models at level 0, see Section \ref{sec:exp-settings}, paragraph Initialization). The training losses and costs are presented in Figure \ref{fig:train-cost}.

\begin{figure}[h]
\vspace{-0.1in}
\centering
\includegraphics[width = 3.4in]{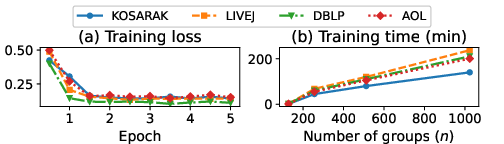}
\caption{Training losses and costs}
\label{fig:train-cost}
\vspace{-0.15in}
\end{figure}

As is clear from Figure \ref{fig:train-cost}(a), on all datasets used in the experiment, the training loss decreases rapidly and the model converges after approximately two epochs, attesting to the efficiency of the model training process. Also, as can be observed from Figure \ref{fig:train-cost}(b), the training cost grows linearly with respect to the number of groups, making \les\ scalable for large datasets. Besides, models at the same level of the cascade framework can be trained in parallel to further reduce the training cost, which is an interesting direction for future investigation.
}
\subsection{PTR vs. Set Representation Techniques}
\label{sec:exp-rep}
We compare PTR with other applicable set representation techniques. More specifically, we choose PCA \cite{jolliffe2016principal}, a widely-used linear embedding method, MDS \cite{de2004sparse}, a representative non-linear embedding approach, {and Binary Encoding \cite{han2011data}, an efficient categorical data embedding technique. We also include the variant of PTR constructed on the first half of the path table (see Section \ref{sec:setrep}), denoted by PTR-half.} Considering the complexity of PCA and MDS, we conduct experiments on sampled KOSARAK (sample ratio of $5\%$). We report the representation construction time of each method and the query answering time using the resulting partitioning results for $k$NN query ($k=10$) and range query ($\delta=0.7$) in Figure \ref{fig:rep}; similar trends are observed on other datasets and queries.
\begin{figure}[h]
\centering
\includegraphics[width = 3.4in]{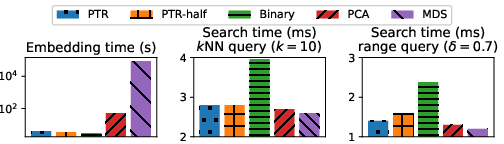}
\caption{Comparison of representation techniques}
\label{fig:rep}
\end{figure}

As can be observed from Figure \ref{fig:rep}, compared with PCA and MDS, PTR incurs much lower embedding time (10 to 20,000 times faster) while results in similar search time; {compared with Binary Encoding and PTR-half, PTR leads to faster query answering with comparable embedding cost. Binary Encoding assigns unique representations to different sets without considering set characteristics (e.g., tokens contained therein), and thus can hardly achieve any Set Separation-Friendly Property. PTR-half, as discussed in Section \ref{sec:setrep}, suffers from the risk that different sets may have common representations, and consequently these (dissimilar) sets are partitioned into the same group as they are not separable in the representation space, and the resulting search time thus is slightly longer than that of PTR.} The major advantage of PTR is that it integrates the Set Separation-Friendly Property introduced in Section \ref{sec:setrep} into set representations by allowing sets consisting of different tokens to be easily separable by axis-aligned hyper-planes in the embedding space, and thus eases the training of the downstream Siamese networks.
\subsection{L2P vs. Algorithmic Approaches}
\label{sec:exp-alg}
{We compare the learning-based partitioning approach, L2P, to the algorithmic methods introduced in Section \ref{sec:heuristic}, namely the graph cut-based method (PAR-G), centroid-based method (PAR-C), divisive clustering method (PAR-D), and {agglomerative clustering method (PAR-A)}, in terms of partitioning cost, including time cost and space cost, and query answering time.

{For PAR-G, we adopt PaToH \cite{ccatalyurek2011patoh}, a graph partitioning tool known to be efficient and performing well, to cut the graph.} We report the cost of different methods in partitioning KOSARAK into 1024 groups and the query answering time for $k$NN with $k=10$ in Figure \ref{fig:heuristic}; similar trends are observed on other datasets and queries. { Note that the partitioning time of L2P includes model training time and inference time (the time required to assign a set into a group), and the partitioning time of PAR-G consists of the $k$NN graph construction time and the graph cut time.} PAR-G is specially optimized for $k=10$ and the construction of its $k$NN graph is accelerated by \les.
\begin{figure}[h]
\vspace{-0.1in}
\centering
\includegraphics[width = 3.4in]{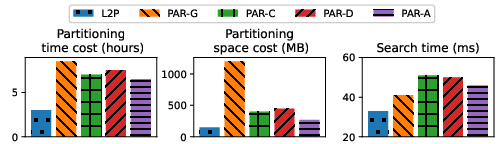}
\vspace{-0.1in}
\caption{Comparison of partitioning methods}
\label{fig:heuristic}
\vspace{-0.1in}
\end{figure}

As depicted in Figure \ref{fig:heuristic}, L2P provides the fastest search while only incurs a small fraction of partitioning time and space cost compared to competitors (saving $80\%$ partitioning time and $99\%$ space compared with PAR-G).
The reason why L2P incurs less partitioning time and space overhead is that, as described in Section \ref{sec:siamese} and Section \ref{sec:exp-settings}, by training the models on a small portion of data, L2P is better positioned to approach the optimal partitioning where the $GPO$ is minimized, while other techniques work on the entire dataset and require  (sometimes repetitively) computing the $GPO$ of arbitrary groups (or pairs) of sets. Besides, only model parameters and the training samples in a mini batch have to be saved in memory for L2P, with minimal storage overhead, while other techniques require materializing a large amount of intermediate partitioning results (and the entire $k$NN graph for PAR-G) in memory, incurring prohibitive space consumption.

By directly optimizing the $GPO$ which integrates the two desired properties of a partitioning with higher pruning efficiency, L2P is able to outperform PAR-G, the objective of which is minimizing the number of edges in the similarity graph crossing different sub-graphs rather than $GPO$. {PAR-C, PAR-D,  and PAR-A, although also aim to optimize the $GPO$, suffer from severe local optimality problems}: a set is moved to a group only if such movement reduces the overall $GPO$, while in many cases movements temporarily increasing the $GPO$ must be allowed to determine a global optimum. For example, let $S_i\in\mathcal{G}_i$, $S_j\in\mathcal{G}_j$ be two sets. Assume that moving $S_i$ to $\mathcal{G}_j$ and moving $S_j$ to $\mathcal{G}_i$ when individually carried out both increase the $GPO$, and consequently $S_i$ remains in $\mathcal{G}_j$ and $S_j$ in $\mathcal{G}_i$. However, swapping $S_i$ and $S_j$ may reduce the overall $GPO$ and thus leads to better partitioning. Such swapping cannot be achieved based on the strategy followed by PAR-C and PAR-D. { Similarly, the strategy of PAR-A does not allow the merge of groups temporarily increasing the overall $GPO$, which however may be necessary in identifying a global optimum.}
}

\subsection{Sensitivity to Number of Groups and $k$}
\label{sec:sensitivitytonandk}
We test the performance of \les\  ~in terms of query processing time on $k$NN queries, varying the number of groups $n$ and the result size $k$. The results are presented in Figure \ref{fig:numgroupresultsize}.
\begin{figure}[h]
\vspace{-0.1in}
\centering
\includegraphics[width = 3.4in]{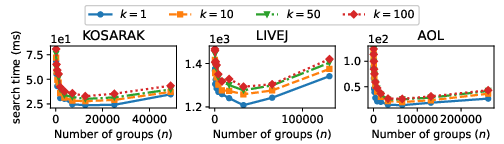}
\vspace{-0.1in}
\caption{Sensitivity to the number of groups and result size}
\label{fig:numgroupresultsize}
\vspace{-0.1in}
\end{figure}

Increasing $n$ accelerates query answering, regardless of the result size. This is because with more groups, as indicated by Equation (\ref{eq:unified-obj}), the overall pruning efficiency of \les\ can be improved, meaning fewer candidates have to be checked. Increasing $n$ benefits search time up to a point. {In particular we observe a diminishing return behavior with respect to search performance as $n$ increases further. The reason is that, with a sufficiently large number of groups, sets are already well-separated, and further increasing $n$ brings no significant change to the pruning efficiency but incurs higher index (TGM) scan cost.} Moreover, search time increases for larger $k$, which is consistent with our analysis in Equation (\ref{eq:range-ub-analysis}), as in general a larger $k$ in $k$NN search is analogous to a smaller $\delta$ in range search and thus the pruning efficiency is lower.

{While determining the optimal number of groups for partitioning is a known NP-hard problem \cite{james2013introduction}, we empirically observe from the experiments that setting the number of groups at approximately $0.5\%|\mathcal{D}|$ leads to the lowest search time, where $|\mathcal{D}|$ is the number of sets in the corresponding dataset.}

\subsection{\les\ vs. Set Similarity Search Baselines}
\label{sec:sssbaselines}
In this section, we compare \les\ with  existing set similarity search approaches to answering $k$NN and range queries in memory-based and disk-based settings respectively. Among tree-based set similarity search approaches to date, DualTrans \cite{DBLP:journals/tkde/ZhangWWX20} provides the fastest query processing. { For inverted index-based methods, we adopt the method proposed in \cite{wang2019leveraging} (denoted by InvIdx), which  yields the state-of-the-art performance for set similarity join tasks. Note that we exclude methods requiring index construction during query time \cite{deng2018overlap,xiao2011efficient,deng2015efficient} as the index construction cost is much higher than the query cost. Since inverted index-based methods are designed for range queries and do not naturally support $k$NN queries, we modify the query answering algorithm of InvIdx for $k$NN queries as follows.
    (1) Given a query set $Q$ and a result size $k$, start with threshold $\delta=1.0$ and use InvIdx to find candidate sets from $\cS$ whose similarity with $Q$ exceeds $\delta$, denoted by $\mathcal{C}$.
    (2) Identify the temporary $k$NN results from $\mathcal{C}$, denoted by $\mathcal{R}_{k}$. If the minimal similarity between any set in $\mathcal{R}_{k}$ and $Q$ exceeds $\delta$, terminate. Otherwise, decrease $\delta$ by $z$, use InvIdx to find candidates sets with the new $\delta$, update $\mathcal{C}$ accordingly, and repeat the step.
    (3) Upon termination, $\mathcal{R}_k$ is guaranteed to be the $k$NN to $Q$, as the similarity between any sets in $\cS\setminus\mathcal{C}$ and $Q$ does not exceed the current $\delta$. The value of $z$ is tuned for faster query answering.
}

In addition, we also include a brute-force approach, i.e., computing the similarity between the query set and all other sets to derive the results, for completeness of comparison.

{In Figure \ref{fig:construct}, we show the index size and index construction time for all methods. It is clear that the indexing structure of \les, namely TGM, is much more lightweight, requiring up to $90\%$ less space than DualTrans and InvIdx. The major time cost of constructing TGM comes from the model training, which however is a preprocessing step incurring only a one-time cost and can be further reduced as discussed in Section~\ref{sec:train-cost}.}
\begin{figure}[t]
\centering
\includegraphics[width = 3.4in]{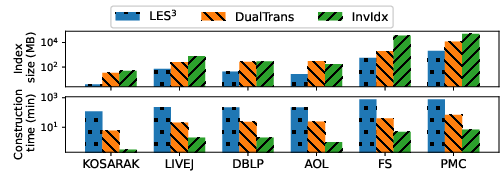}
\vspace{-0.1in}
\caption{Index size and construction time}
\label{fig:construct}
\vspace{-0.05in}
\end{figure}

In Figure \ref{fig:compare-memory-all} we compare the performance of the four methods in a memory-based setting. We observe that \les\ outperforms competitors for both $k$NN queries and range queries, accelerating the query answering by 2 to 20 times. DualTrans incurs longer search time as it uses an R-tree to organize all sets, with each set being represented with a $d$-dimensional vector ($d$ can be tuned for faster pruning). When the value of $d$ is small, sets containing different tokens cannot be clearly separated based on their representations, while when the value of $d$ is large, using R-tree to organize the vectors incurs high overlap between the bounding boxes of nodes on the R-tree, as previous research indicates \cite{indyk1998approximate}. Besides, scanning the R-tree is expensive, which is not worthwhile considering that set similarity (e.g., Jaccard similarity) can usually be computed efficiently. {While InvIdx provides comparable performance with \les\ for range queries with large $\delta$, it incurs greater search latency for $k$NN queries, especially when the average set size is large (e.g., on KOSARAK and LIVEJ). The reason is that, with InvIdx filtering operations need to be repeated for each candidate set (or multiple candidates with some common characteristics), and larger set size and $k$NN queries both enlarge the number of candidates, leading to sub-par query performance.}

\begin{figure*}[t]
\centering
\includegraphics[width = 6.9in]{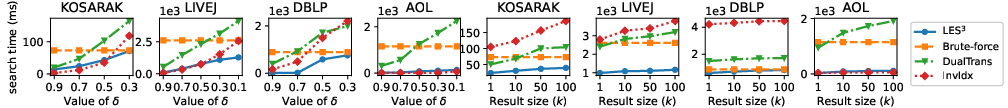}
\caption{Comparison to baselines in memory-based settings for range queries (left) and $k$NN queries (right)}
\label{fig:compare-memory-all}
\vspace{-0.1in}
\end{figure*}

\begin{figure}[t]
\centering
\includegraphics[width = 3.4in]{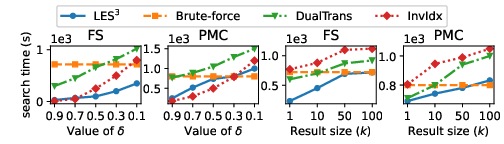}
\caption{Comparison to baselines in disk-based settings for range queries (left) and $k$NN queries (right)}
\label{fig:compare-disk-k}
\end{figure}

In contrast,  we use TGM to compute the upper bounds between a query set and a group of sets; obtaining all bounds requires only $O(\vert S\vert*\vert\mathcal{G}\vert)$ time, which is relatively cheap. Although the search time of \les\ increases for range query as $\delta$ decreases, \les\ provide much faster query answering under a wide range of $\delta$.


{We compare the performance of the four methods in the disk-based setting in Figure \ref{fig:compare-disk-k}. Note that for DualTrans and InvIdx, only the part of the index that is necessary to the query answering, such as R-nodes on the search path and inverted indexes related to the query set, is retrieved into memory to reduce I/O cost.} We observe that \les\ generally provides faster search compared with competitors, accelerating the query answering by 2 to 10 times. The reasons why \les\ incurs lower search time are: (1) Sets sharing no or very few common tokens with the query set can be easily pruned without being retrieved into memory; and (2) Since sets in the same group are checked jointly during the searching process; materializing a group of sets continuously on disk minimizes the data transfer delay. { DualTrans and InvIdx, on the contrary, incur longer search latency and are outperformed by the Brute-force method for a wide range of $k$ and $\delta$. Besides the drawbacks discussed above in the memory-based setting, the search strategies of DualTrans and InvIdx incur repetitive retrieval of data with random disk access, which results in higher I/O cost (more pages retrieved, higher seek and rotation overhead, etc.), making them less efficient in the disk-based setting.}
\subsection{TGM vs. HTGM}
\label{sec:tgm-htgm}
We evaluate the performance of TGM and HTGM to determine whether building a hierarchical index pays off. Intuitively, whether it benefits the query processing largely depends on the similarity distribution. For example, in cases where very few sets share common tokens, one can prune a large number of candidates using the matrices at the first few levels of HTGM, avoiding scanning the larger matrices at finer levels. However, in cases where most sets are similar, the small matrices at the first few levels of HTGM may provide no pruning efficiency at all. We assume that the similarity between sets in $\cS$ can be modeled by a power-law distribution $P[sim=v]\sim v^{-\alpha}$, where $P[sim=v]$ denotes the probability that the similarity between any two sets is $v$, $v\in[0,1]$, $\alpha\in[1,\infty)$. We generate multiple synthetic databases consisting of 20,000 sets and 20,000 tokens each, by varying the value of $\alpha$. We train a cascade model with 9 levels (including level 0). We use the partitioning results at level 8 (256 groups) to build the TGM, and use the partitioning results at level 5 (32 groups) and level 8 to build the HTGM. We compare HTGM and TGM from two aspects. First, the index access cost, measured by the number of columns in the HTGM or TGM that are checked when processing the query. Second, the computational cost, measured by the number of similarity calculations. We measure the ratio of cost between HTGM and TGM, and the results are shown in Figure~\ref{fig:tgm-vs-htgm}. It is evident that HTGM outperforms TGM when the value of $\alpha$ is large, i.e., most sets are dissimilar. This is in line with the discussions in Section \ref{sec:tgm-htgm}. 

\noindent\begin{minipage}{.24\textwidth}
\centering
\includegraphics[width=\textwidth]{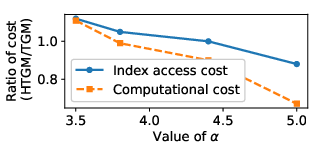}
\captionof{figure}{TGM vs. HTGM}
\label{fig:tgm-vs-htgm}            
\end{minipage}
\noindent\begin{minipage}{.24\textwidth}
\centering
\includegraphics[width=\textwidth]{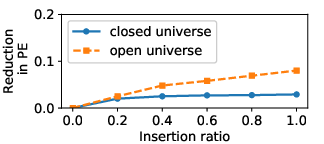}
\captionof{figure}{Handling updates}
\label{fig:updates}            
\end{minipage}
\subsection{Handling Updates}
\label{sec:update}
We evaluate the performance of the proposed approach under updates. Two cases are considered: (1) {\em closed universe}, meaning the new sets to be inserted contain only tokens from the original database, and (2) {\em open universe}, where the new sets may contain previously unseen tokens. Let $\mathcal{D}$ be the original database, $\mathcal{D}^{closed}$ be the collection of new sets to be inserted under a closed universe, and $\mathcal{D}^{open}$ be the collection of new sets to be inserted under an open universes. For the experiment, we set insertion ratio ($\vert\mathcal{D}^{closed}\vert/\vert\mathcal{D}\vert$ and $\vert\mathcal{D}^{open}\vert/\vert\mathcal{D}\vert$) in range $[0,1]$, and half of the tokens in $\mathcal{D}^{open}$ are from $\mathcal{D}$ and half are new. We compute the decrease in pruning efficiency after insertion compared to obtaining a partitioning from scratch (namely running L2P) on $\mathcal{D}\cup \mathcal{D}^{closed}$ or $\mathcal{D}\cup \mathcal{D}^{open}$ (referred to as {\em re-build}). We give the results for $k$NN query with $k=10$ on KOSARAK in Figure \ref{fig:updates}; the experiments on the other datasets show similar trends.

Figure \ref{fig:updates} depicts the percentage of pe reduction compared to re-build. The pruning efficiency decreases slightly as more new sets are inserted into the database. Insertions under an open universe have a higher impact on performance. The reason is that the tokens from the same universe mainly follow a similar distribution and the partition results obtained on the original data are still sufficient, while there is no prior knowledge regarding the distribution of new tokens. We observe, however, that the overall pruning efficiency is resistant to insertions (experiencing a decrease by at most $8\%$), which attests to the robustness of the proposed approach.

\section{Related Work}
The problem of processing set similarity queries, including set similarity search \cite{kim2012efficient,zhang2017efficient,DBLP:journals/tkde/ZhangWWX20} and set similarity joins \cite{deng2018overlap,deng2019lcjoin,fernandez2019lazo,mann2016empirical,wang2019leveraging}, has attracted remarkable research interest recently. Zhang et al. \cite{zhang2017efficient,DBLP:journals/tkde/ZhangWWX20} propose to transform sets into scalars or vectors with the relative distance between sets preserved, and organize the transformed sets with B+-trees or R-trees, which facilitate the use of tree-based branch-and-bound algorithms for similarity search. The major drawback of their work is that, as shown in the experiments, the tree structure can easily grow larger than the original data and thus using the index for filtering incurs a significant cost, especially when the index and the data are stored externally. Most prior research in the area of set similarity join focuses on threshold-join queries and follows the filter-and-verify framework. In the filter step, existing methods mainly adopt (1) prefix-based filters \cite{bouros2012spatio,wang2012can,xiao2011efficient}, based on the observation that if the similarity between two sets exceeds $\delta$, then they must share common token(s) in their prefixes of length $m$; and (2) partition-based filters \cite{arasu2006efficient,deng2015efficient,deng2018overlap,xiao2009top}, which partition a set into several subsets so that two sets are similar only if they share a common subset. 
Set similarity queries in distributed environments \cite{das2014clusterjoin,metwally2012v} and approximate queries \cite{satuluri2011bayesian,schelter2016tracking} are beyond the scope of this paper but represent promising directions for further investigation.

Indexing is an important and well-studied problem in data management and recent works have utilized machine learning to learn a CDF or to partition the data space for traditional database indexing \cite{li2020bindex,yue2020analysis,lang2020tree,li2020lisa,ding2020alex,nathan2020learning,DBLP:conf/iclr/SablayrollesDSJ19,Dong2020Learning}. In this paper, we complement recent work by studying the applicability of machine learning techniques to assist index construction for 
set similarity search problems. Our results show that the proposed methods offer vast advantages over traditional techniques.


Embedding sets and other entities consisting of discrete elements has been well-studied. The most natural way to represent such data types is $n$-hot encoding, but the resulting vectors are often very long and sparse. Dimensionality reduction techniques are used to compress the encoding vectors with different focuses: maximizing variances \cite{pearson1901liii}, preserving distances \cite{borg2003modern}, solving the crowding problem \cite{maaten2008visualizing}, etc. Recent advances in document embedding, e.g., word2vec \cite{mikolov2013efficient}, BERT \cite{devlin2018bert}, also provide new perspectives to construct representations of sets. Compared to these methods, the PTR proposed in Section \ref{sec:setrep} utilizes a very efficient method to produce relatively short representations and is optimized for the specific problem at hand. 

\section{Conclusions}
In this paper, we have studied the problem of {\em exact set similarity search}, and designed \les, a filter-and-verify approach for efficient query processing. Central to our proposal is TGM, a simple yet effective structure that strikes a balance between index access cost and effectiveness in pruning candidate sets. We have revealed the desired properties of optimal partitioning in terms of pruning efficiency under the uniform token distribution assumption. We develop a learning-based approach, L2P, utilizing a cascade of Siamese networks to identify partitions. A  novel set representation method, PTR, is developed to cater to the requirements of network training. The experimental results have demonstrated the superiority of \les\ over other applicable approaches. 


\balance
\bibliographystyle{ACM-Reference-Format}
\bibliography{reference}

\end{document}